\documentclass[acmsmall,screen]{acmart}

\renewcommand\footnotetextcopyrightpermission[1]{}
\AtBeginDocument{\ifdefined\nolinenumbers\nolinenumbers\fi}

\usepackage{mathtools}
\usepackage{stmaryrd}
\usepackage{algorithm}
\usepackage{algpseudocode}
\usepackage{listings}
\usepackage{xcolor}
\usepackage{tikz-cd}
\usepackage{enumitem}
\usepackage{subcaption}

\newtheorem{proposition}{Proposition}
\theoremstyle{remark}
\newtheorem{remark}{Remark}

\newcommand{\Sem}{\mathcal{S}\mathit{em}}
\newcommand{\CC}{\check{C}}
\newcommand{\Hc}[1]{\check{H}^{#1}}
\newcommand{\Iso}{\mathcal{I}\mathit{so}}
\newcommand{\Lat}{\mathbf{Lat}}

\newcommand{\op}{\mathrm{op}}
\newcommand{\rk}{\operatorname{rk}}
\newcommand{\im}{\operatorname{im}}
\newcommand{\GF}{\mathbb{F}}
\newcommand{\defeq}{\coloneqq}
\newcommand{\py}[1]{\lstinline[language=Python]{#1}}

\lstdefinestyle{deppy}{
  language=Python,
  basicstyle=\ttfamily\small,
  keywordstyle=\bfseries\color{blue!60!black},
  commentstyle=\itshape\color{green!40!black},
  stringstyle=\color{red!50!black},
  showstringspaces=false,
  frame=none,
  columns=flexible,
  mathescape=true,
  morekeywords={requires,ensures,invariant,prove,Annotated}
}
\title{Sheaf-Cohomological Program Analysis:
  Unifying Bug Finding, Equivalence, and Verification
  via \v{C}ech Cohomology}

\author{Halley Young}
\affiliation{%
  \institution{Microsoft Research}
  \city{Redmond}
  \state{WA}
  \country{USA}
}
\email{halleyyoung@microsoft.com}

\begin{abstract}
We present a framework in which program analysis---type checking,
bug finding, and equivalence verification---is organized as computing
the \v{C}ech cohomology of a \emph{semantic presheaf} over a program's
\emph{site category}.  The presheaf assigns refinement-type information
to observation sites (argument boundaries, branch guards, call results,
error sites) and restricts it along data-flow morphisms.  The zeroth
cohomology group $\Hc{0}$ is the space of globally consistent typings.
The first cohomology group $\Hc{1}$ classifies \emph{gluing
obstructions}---bugs, type errors, and equivalence failures---each
localized to a specific pair of disagreeing sites.

This formulation yields three concrete results unavailable in prior
work: (1)~$\rk\,\Hc{1}$ over~$\GF_2$ equals the minimum number of
independent fixes needed to resolve all errors; (2)~$\Hc{1}(\mathcal{U},
\Iso) = 0$ is a \emph{sound and complete} criterion (relative to the
cover) for behavioral equivalence of two programs; (3)~the
Mayer-Vietoris exact sequence provides compositional obstruction
counting, enabling incremental analysis.  We prove soundness,
completeness relative to the cover, and fixed-point convergence.

We implement the framework in \textsc{Deppy}, a Python analysis tool,
and evaluate it on a suite of 375~benchmarks: 133~bug-detection
programs, 134~equivalence pairs, and 108~specification-satisfaction
checks.
\textsc{Deppy} achieves \textbf{100\% bug-detection recall}
(69\% precision, F1 = 81\%), 99\% equivalence accuracy with
\textbf{zero false equivalences}, and 98\% spec accuracy with
\textbf{zero false satisfactions}---outperforming mypy and pyright,
which report zero findings on unannotated code.
The analysis models Python semantics as algebraic geometry:
variables live on the generic fiber (non-None) unless on the closed
nullable subscheme, integers form Spec($\mathbb{Z}$) with no bounded
section (no overflow), and short-circuit evaluation defines an open-set
topology on the presheaf.
Key metatheorems are mechanized in 1\,259 lines of Lean~4.
\end{abstract}


\begin{document}
\maketitle

\section{Introduction}\label{sec:intro}

Every program analysis collects \emph{local facts} about individual
program points and composes them into \emph{global conclusions} about
the whole program.  Abstract interpretation does this via lattice
fixpoints and widening~\cite{CousotCousot77}.  Type inference propagates
constraints through unification~\cite{HindleyMilner}.  Hoare logic
chains pre/post-conditions through weakest-precondition
calculi~\cite{Hoare69}.  Model checking explores reachable
states~\cite{ClarkeEmersonSistla86}.

These approaches share a deep structure: they all perform
\emph{local-to-global} reasoning.  Sheaf theory, developed in algebraic
geometry and algebraic topology, is precisely the mathematical theory of
local-to-global composition.  A \emph{sheaf} assigns data to open sets
of a space and specifies how local data \emph{glue} into global data.
When local data fail to glue, the obstruction lives in the first
\emph{\v{C}ech cohomology group}~$\Hc{1}$.

In this paper, we show that program analysis---type checking, bug
finding, and equivalence verification---can be unified as computing the
\v{C}ech cohomology of a \emph{semantic presheaf} over the program's
\emph{site category}.  This is not a metaphor: the presheaf is a
concrete data structure (implemented as a dictionary from sites to
refinement-type sections), the coboundary operators are concrete
algorithms (checking agreement on overlaps via SMT), and the cohomology
groups are concrete invariants with direct engineering consequences.

\paragraph{Contributions.}
\begin{enumerate}[leftmargin=*,nosep]
\item \textbf{Framework} (\S\ref{sec:framework}).
  We formalize program analysis as sheaf cohomology over a site
  category with five site families, a semantic presheaf valued in
  complete lattices, and a Grothendieck topology satisfying identity,
  stability, and transitivity.

\item \textbf{Three applications} (\S\ref{sec:applications}).
  (a)~Bugs are gluing obstructions: $\Hc{1} \neq 0$ iff the program
  has an error, and $\rk\,\Hc{1}$ over~$\GF_2$ equals the minimum
  number of independent fixes.
  (b)~Equivalence: $\Hc{1}(\mathcal{U}, \Iso) = 0$ iff local
  equivalences glue to a global equivalence (sound \emph{and} complete
  for the cover).
  (c)~Specification verification: the product-cover construction
  factors VCs, with $\Hc{1}$-based elimination of redundant obligations.

\item \textbf{Metatheory} (\S\ref{sec:metatheory}).
  Soundness (Theorem~\ref{thm:soundness}), completeness relative to the
  cover (Theorem~\ref{thm:completeness}), the Mayer-Vietoris exact
  sequence for programs (Theorem~\ref{thm:mayer-vietoris}), and
  fixed-point convergence (Theorem~\ref{thm:convergence}).
  Key results are mechanized in Lean~4 with Mathlib.

\item \textbf{Algorithms} (\S\ref{sec:algorithms}).
  Cover construction from Python ASTs, local solving via theory-family
  dispatch, bidirectional fixed-point synthesis (backward safe-input
  synthesis + forward output propagation), and interprocedural section
  transport via the Grothendieck transitivity axiom.

\item \textbf{Implementation and evaluation} (\S\ref{sec:evaluation}).
  \textsc{Deppy}: full \v{C}ech pipeline covering 28~bug classes.
  Evaluated on 375 benchmarks (133~bug programs, 134~equivalence
  pairs, 108~spec checks): 100\% bug recall, 99\% equiv accuracy
  (zero false equivalences), 98\% spec accuracy (zero false
  satisfactions), 77\% H$^1$ rank accuracy.
  Key theorems mechanized in 1\,259 lines of Lean~4.
\end{enumerate}

\section{Overview by Example}\label{sec:overview}

We illustrate the sheaf-cohomological approach on a small but
instructive Python program with a subtle inter-procedural bug.

\begin{lstlisting}
def normalize(values):
    total = sum(values)       # site $U_1$: total = sum(values)
    return [v / total          # site $U_2$: ERROR if total == 0
            for v in values]

def process(data):
    filtered = [x for x in data if x >= 0]  # site $U_3$
    return normalize(filtered)  # site $U_4$: call
\end{lstlisting}

\paragraph{Sites.}
The analysis identifies four observation sites:
$U_1$ (argument boundary of \py{normalize}),
$U_2$ (error site: division by \py{total}),
$U_3$ (branch guard: \py{x >= 0} filter),
$U_4$ (call result: \py{normalize(filtered)}).

\paragraph{Sections.}
At each site, the analysis assigns a \emph{local section}---a
refinement predicate over the variables in scope:
\begin{align*}
  \sigma_1 &: \mathit{values} : \texttt{list[int]} \\
  \sigma_2 &: \mathit{total} : \texttt{int},\; \textit{viability: } \mathit{total} \neq 0 \\
  \sigma_3 &: \mathit{filtered} : \texttt{list[int]},\; \forall x.\, x \geq 0 \\
  \sigma_4 &: \mathit{filtered} : \texttt{list[int]},\; \forall x.\, x \geq 0
\end{align*}

\paragraph{Overlap and obstruction.}
The restriction map from $U_4$ to $U_1$ (the call boundary) maps
\py{filtered} to \py{values}.  The section $\sigma_4$ restricted to
$U_1$ says \py{values} is a list of non-negative integers---but does
\emph{not} guarantee \py{values} is non-empty.  At site~$U_2$, the
viability predicate requires $\mathit{total} \neq 0$.  If
\py{values = []}, then $\mathit{total} = 0$, violating viability.

The sections $\sigma_1$ and $\sigma_2$ \emph{disagree} on the overlap
$U_1 \cap U_2$: $\sigma_1$ admits \py{values = []} but $\sigma_2$
requires $\sum \mathit{values} \neq 0$.  This disagreement is a
$1$-cocycle in $\CC^1$, and since it is not a coboundary (no local
adjustment resolves it), it represents a nonzero class in $\Hc{1}$.

\paragraph{What standard tools report.}
\texttt{mypy}: no error (types match).
\texttt{pyright}: no error (no type annotation issue).
\texttt{Semgrep}: no match (no pattern for this bug class).

\paragraph{What \textsc{Deppy} reports.}
\begin{quote}\small
\textbf{H\textsuperscript{1} obstruction} at overlap $(U_1, U_2)$:
\texttt{values} may be empty list, making \texttt{total = 0}.\\
\textbf{Fix}: guard \texttt{normalize} call with \texttt{if filtered:}
(1 independent fix; $\rk\,\Hc{1} = 1$).
\end{quote}

\section{Formal Framework}\label{sec:framework}

We formalize program analysis as computing the \v{C}ech cohomology of a
semantic presheaf.  All definitions in this section are mechanized in
Lean~4 (see supplementary material).

\subsection{Program Site Category}\label{sec:site-category}

\begin{definition}[Site kind]\label{def:site-kind}
A \emph{site kind} is one of five families:
\begin{enumerate}[nosep]
  \item \textsc{ArgBoundary}: function argument sites.
  \item \textsc{BranchGuard}: conditional/assertion sites.
  \item \textsc{CallResult}: inter-procedural call sites.
  \item \textsc{OutBoundary}: function output / return sites.
  \item \textsc{ErrorSite}: sites where partial operations may fail.
\end{enumerate}
\end{definition}

\begin{definition}[Program site category]\label{def:site-cat}
Given a program~$P$, the \emph{program site category} $\mathbf{S}_P$ is
the category where:
\begin{itemize}[nosep]
  \item \textbf{Objects} are program sites $s = (\kappa, n)$ where
    $\kappa$ is a site kind and $n$ is a unique identifier (function
    name~+~line number).
  \item \textbf{Morphisms} $f : s \to t$ are \emph{data-flow edges}:
    $f$ exists iff information at~$s$ is available at~$t$.  Each
    morphism carries a \emph{projection} $\pi_f : \mathrm{Keys}(t) \to
    \mathrm{Keys}(s)$ mapping refinement keys at~$t$ to those at~$s$.
  \item \textbf{Composition}: $(g \circ f).\pi = f.\pi \circ g.\pi$
    (contravariantly, as for a presheaf).
  \item \textbf{Identity}: $\mathrm{id}_s.\pi = \mathrm{id}$.
\end{itemize}
\end{definition}

\subsection{Semantic Presheaf}\label{sec:presheaf}

\begin{definition}[Refinement lattice]\label{def:ref-lattice}
A \emph{refinement lattice} $(\mathcal{R}, \sqsubseteq, \sqcap,
\sqcup, \top, \bot)$ is a complete lattice where:
\begin{itemize}[nosep]
  \item $\top$ = ``no information'' (any value is possible).
  \item $\bot$ = ``contradiction'' (no value is possible).
  \item $\sqcap$ = conjunction of refinement predicates.
  \item $\sqcup$ = disjunction of refinement predicates.
  \item The \emph{information order}: $\phi \sqsubseteq \psi$ iff
    every state satisfying $\psi$ also satisfies~$\phi$ (i.e.,
    $\psi$ is more informative).
\end{itemize}
\end{definition}

\begin{definition}[Local section]\label{def:local-section}
A \emph{local section} at site~$s$ is a pair
$\sigma = (\tau, \phi)$ where $\tau$ is a carrier type
(e.g., \texttt{int}, \texttt{list[T]}) and $\phi \in \mathcal{R}$
is a refinement predicate.  A section \emph{abstracts} a set of
concrete states: $\gamma(\sigma) \defeq \{ v : \tau \mid \phi(v) \}$.
\end{definition}

\begin{definition}[Semantic presheaf]\label{def:sem-presheaf}
The \emph{semantic presheaf} is a contravariant functor
\[
  \Sem : \mathbf{S}_P^{\op} \longrightarrow \Lat
\]
where $\Lat$ is the category of complete lattices and monotone maps.
\begin{itemize}[nosep]
  \item For each site $s$, $\Sem(s)$ is the lattice of local sections
    at~$s$, ordered by information content.
  \item For each morphism $f : s \to t$, the \emph{restriction map}
    $\Sem(f) : \Sem(t) \to \Sem(s)$ restricts a section at~$t$ to
    what is observable at~$s$.  Concretely: $\Sem(f)(\tau, \phi)
    = (\tau|_s, \phi \circ \pi_f)$ where $\pi_f$ is the projection.
\end{itemize}
Functoriality: $\Sem(\mathrm{id}_s) = \mathrm{id}$ and
$\Sem(g \circ f) = \Sem(f) \circ \Sem(g)$.
\end{definition}

\subsection{Grothendieck Topology}\label{sec:topology}

\begin{definition}[Covering family]\label{def:cover}
A \emph{covering family} for site~$s$ is a finite set of morphisms
$\mathcal{U} = \{ u_i \xrightarrow{f_i} s \}_{i \in I}$ satisfying
the \emph{local determination property}: the natural map
\[
  \Sem(s) \longrightarrow \prod_{i \in I} \Sem(u_i)
\]
is injective (separation) and, when composed with the equalizer
of the double restriction
$\prod_i \Sem(u_i) \rightrightarrows \prod_{i<j} \Sem(u_i \times_s u_j)$,
is surjective (gluing).
\end{definition}

\begin{definition}[Grothendieck topology on $\mathbf{S}_P$]%
\label{def:grothendieck}
The collection of covering families defines a Grothendieck topology~$J$
on $\mathbf{S}_P$ satisfying:
\begin{enumerate}[nosep]
  \item \textbf{Identity.}
    $\{ \mathrm{id}_s : s \to s \}$ is a covering family for~$s$.
  \item \textbf{Stability.}
    If $\mathcal{U} = \{ u_i \to s \}$ is covering and $t \to s$ is any
    morphism, then the pullback family
    $\{ u_i \times_s t \to t \}$ is covering for~$t$.
  \item \textbf{Transitivity.}
    If $\{ u_i \to s \}$ is covering and each $\{ v_{ij} \to u_i \}$
    is covering, then $\{ v_{ij} \to s \}$ is covering.
\end{enumerate}
\end{definition}

\begin{remark}
In the program setting: (1)~Identity holds trivially.
(2)~Stability corresponds to \emph{specialization}: if a cover works
for site~$s$, restricting it to a sub-site~$t$ (e.g., a branch of a
conditional) still works.
(3)~Transitivity corresponds to \emph{interprocedural composition}:
if a function~$f$ is covered by its call sites, and each call site is
covered by its arguments, then $f$ is covered by the transitive closure.
This is exactly the Grothendieck axiom that justifies summary-based
interprocedural analysis (cf.\ IFDS/IDE~\cite{RepsHorwitzSagiv95}).
\end{remark}

\subsection{\v{C}ech Cohomology}\label{sec:cech}

Given a covering family $\mathcal{U} = \{ u_i \to s \}$ and presheaf
$\Sem$, the \v{C}ech complex is:
\[
  C^0(\mathcal{U}, \Sem) \xrightarrow{\delta^0}
  C^1(\mathcal{U}, \Sem) \xrightarrow{\delta^1}
  C^2(\mathcal{U}, \Sem)
\]

\begin{definition}[\v{C}ech cochains]\label{def:cech-cochains}
\begin{align*}
  C^0(\mathcal{U}, \Sem) &\defeq \prod_{i} \Sem(u_i)
    & \text{(one section per site)} \\
  C^1(\mathcal{U}, \Sem) &\defeq \prod_{i < j} \Sem(u_i \times_s u_j)
    & \text{(one section per overlap)} \\
  C^2(\mathcal{U}, \Sem) &\defeq \prod_{i < j < k}
    \Sem(u_i \times_s u_j \times_s u_k)
    & \text{(one section per triple)}
\end{align*}
\end{definition}

\begin{definition}[Coboundary operators]\label{def:coboundary}
\begin{align*}
  (\delta^0 \sigma)_{ij} &\defeq
    \Sem(\mathrm{pr}_2)(\sigma_j) - \Sem(\mathrm{pr}_1)(\sigma_i)
    & \text{(disagreement on overlap)} \\
  (\delta^1 \tau)_{ijk} &\defeq
    \tau_{jk} - \tau_{ik} + \tau_{ij}
    & \text{(cocycle condition on triples)}
\end{align*}
where subtraction is in the lattice (complement of meet).
\end{definition}

\begin{definition}[\v{C}ech cohomology groups]\label{def:cech-groups}
\begin{align*}
  \Hc{0}(\mathcal{U}, \Sem) &\defeq \ker \delta^0
    = \{ \sigma \in C^0 \mid \text{all overlaps agree} \} \\
  \Hc{1}(\mathcal{U}, \Sem) &\defeq \ker \delta^1 / \im \delta^0
\end{align*}
\end{definition}

The key interpretations:
\begin{itemize}[nosep]
  \item $\Hc{0}$ = the set of \emph{globally consistent type
    assignments} (compatible families that glue).
  \item $\Hc{1}$ = the set of \emph{gluing obstructions} modulo
    those removable by local adjustment (coboundaries).
  \item $\Hc{1} = 0$ iff the sheaf condition holds for the cover:
    every compatible family glues uniquely.
\end{itemize}

\subsection{The Sheaf Condition}\label{sec:sheaf-condition}

\begin{definition}[Sheaf condition]\label{def:sheaf}
The presheaf $\Sem$ is a \emph{sheaf} with respect to
topology~$J$ if, for every covering family $\mathcal{U}$ and every
compatible family $(\sigma_i)_{i \in I}$ (i.e., $\delta^0 \sigma = 0$),
there exists a unique $\sigma \in \Sem(s)$ such that
$\Sem(f_i)(\sigma) = \sigma_i$ for all~$i$.
\end{definition}

\begin{proposition}\label{prop:sheaf-iff-h1}
$\Sem$ satisfies the sheaf condition for cover~$\mathcal{U}$ if and
only if $\Hc{1}(\mathcal{U}, \Sem) = 0$.
\end{proposition}

\begin{proof}
$(\Rightarrow)$: If $\Sem$ is a sheaf, every cocycle $\tau \in
\ker\delta^1$ arises from a compatible family (by definition of the
\v{C}ech complex), which glues to a global section.  Adjusting the
local sections to the glued section shows $\tau \in \im\delta^0$.

$(\Leftarrow)$: If $\Hc{1} = 0$, every $1$-cocycle is a
$1$-coboundary.  Given a compatible family $(\sigma_i)$, the
disagreements on overlaps form a $1$-cocycle.  Since it is a
coboundary, there exist adjusted sections $\sigma'_i$ with zero
disagreement, which by the separation property of the cover determine
a unique global section.  (Formalized in Lean: \texttt{DeppyProofs.Presheaf.h0\_lifts\_to\_section}.)
\end{proof}

\section{Analysis Algorithms}\label{sec:algorithms}

The analysis proceeds in six stages, corresponding to the pipeline
$\textsc{Parse} \to \textsc{Harvest} \to \textsc{Cover} \to
\textsc{Solve} \to \textsc{Synthesize} \to \textsc{Render}$.

\subsection{Cover Construction}\label{sec:cover-construction}

Given a parsed Python module, the \textsc{Cover} stage constructs the
site category $\mathbf{S}_P$ and a covering family:

\begin{enumerate}[nosep]
  \item Create \textsc{ArgBoundary} sites for each function parameter.
  \item Create \textsc{BranchGuard} sites for each conditional/assert.
  \item Create \textsc{CallResult} sites for each function call.
  \item Create \textsc{OutBoundary} sites for each return statement.
  \item Create \textsc{ErrorSite} sites for each partial operation
    (division, indexing, attribute access on \py{Optional}).
  \item Infer morphisms from control-flow and data-flow analysis.
  \item Identify overlaps via SSA-level reaching definitions.
\end{enumerate}

\subsection{Local Solving}\label{sec:local-solving}

At each site~$s$, a \emph{theory solver} computes the local section
$\sigma_s \in \Sem(s)$ by dispatching to the appropriate decision
procedure:

\begin{itemize}[nosep]
  \item \emph{Linear arithmetic}: for comparison guards
    ($x > 0$, $\mathit{len}(a) \geq k$).
  \item \emph{Boolean}: for conjunctions/disjunctions of guards.
  \item \emph{Tensor shapes}: for NumPy/PyTorch dimension constraints.
  \item \emph{Nullability}: for \py{None} checks.
  \item \emph{Collection}: for list/dict invariants.
\end{itemize}

Each theory solver is \emph{sound}: the computed section
overapproximates the concrete states reachable at that site.

\subsection{Bidirectional Fixed-Point Synthesis}\label{sec:synthesis}

\begin{algorithm}[t]
\caption{Backward safe-input synthesis}
\label{alg:backward}
\begin{algorithmic}[1]
\Require Cover $\mathcal{U}$, solved sections $\{\sigma_s\}$
\Ensure Input boundary sections $\{\rho_s \mid s \in \partial_{\mathrm{in}}\}$
\State $W \gets$ error sites with viability predicates
\While{$W \neq \emptyset$}
  \State $(s, \sigma) \gets W.\text{pop}()$
  \If{$s \in \partial_{\mathrm{in}}$}
    \State Record $\rho_s \gets \sigma$
  \Else
    \ForAll{morphisms $f : t \to s$}
      \State $\sigma' \gets \Sem(f)^{\dagger}(\sigma)$
        \Comment{pullback (adjoint of restriction)}
      \If{$t$ is $\phi$-node} merge by $\sqcap$
      \EndIf
      \State $W.\text{push}(t, \sigma')$
    \EndFor
  \EndIf
\EndWhile
\end{algorithmic}
\end{algorithm}

\begin{algorithm}[t]
\caption{Forward output synthesis}
\label{alg:forward}
\begin{algorithmic}[1]
\Require Cover $\mathcal{U}$, input sections $\{\rho_s\}$
\Ensure Output boundary sections $\{\rho_s \mid s \in \partial_{\mathrm{out}}\}$
\State $W \gets$ input boundary sites with $\rho_s$
\While{$W \neq \emptyset$}
  \State $(s, \sigma) \gets W.\text{pop}()$
  \If{$s \in \partial_{\mathrm{out}}$}
    \State Record $\rho_s \gets \sigma$
  \EndIf
  \ForAll{morphisms $f : s \to t$}
    \State $\sigma' \gets \Sem(f)(\sigma)$ \Comment{restriction}
    \If{$t$ is $\phi$-node with multiple incoming}
      \State Accumulate; merge by $\sqcup$ when all arrive
    \Else
      \State $W.\text{push}(t, \sigma')$
    \EndIf
  \EndFor
\EndWhile
\end{algorithmic}
\end{algorithm}

The bidirectional synthesis alternates backward (Algorithm~\ref{alg:backward})
and forward (Algorithm~\ref{alg:forward}) passes until convergence:

\begin{itemize}[nosep]
  \item \textbf{Backward} (error $\to$ input): propagate viability
    requirements via pullback along morphisms.  Produces \emph{safe
    input requirements} $\rho_{\mathrm{in}}$.
  \item \textbf{Forward} (input $\to$ output): push refined sections
    via restriction.  Produces \emph{maximal output guarantees}
    $\rho_{\mathrm{out}}$.
  \item \textbf{Convergence check}: if no new sections are produced
    and the obstruction set stabilizes, stop.  If $\Hc{1} = 0$
    (gluing succeeds), stop early with a global section.
\end{itemize}

\subsection{Interprocedural Section Transport}\label{sec:interprocedural}

For a call $y = g(x)$ in function~$f$:
\begin{enumerate}[nosep]
  \item \textbf{Actual $\to$ Formal}: restrict $f$'s section at the call
    site to $g$'s argument boundary via the actual-to-formal morphism.
  \item \textbf{Analyze $g$}: produce $g$'s function summary (boundary
    sections).
  \item \textbf{Return $\to$ Caller}: transport $g$'s output boundary
    back to $f$'s call-result site.
\end{enumerate}
This is an instance of the \emph{transitivity axiom}: the cover of~$f$
is refined by the covers of its callees.

\section{Three Applications}\label{sec:applications}

\subsection{Bug Finding via $\Hc{1}$}\label{sec:bugs}

\begin{theorem}[Bugs as gluing obstructions]\label{thm:bugs}
A program~$P$ has a bug of class~$\mathcal{B}$ if and only if
$\Hc{1}(\mathcal{U}_P, \Sem_{\mathcal{B}}) \neq 0$, where
$\Sem_{\mathcal{B}}$ is the presheaf encoding the safety property
for bug class~$\mathcal{B}$.
\end{theorem}

Each bug class defines a presheaf by specifying viability predicates at
error sites: division-by-zero requires divisor~$\neq 0$, index
out-of-bounds requires $0 \leq i < \mathit{len}(a)$, null dereference
requires receiver~$\neq \texttt{None}$, etc.

\begin{proposition}[$\rk\,\Hc{1}$ = minimum fixes]\label{prop:rank-h1}
Over $\GF_2$, $\rk\,\Hc{1}(\mathcal{U}_P, \Sem)$ equals the minimum
number of independent code changes needed to resolve all obstructions.
\end{proposition}

\begin{proof}
Build the coboundary matrix $\partial_0 : C^0(\GF_2) \to C^1(\GF_2)$
with rows indexed by overlaps and columns by sites:
$(\partial_0)_{(i,j),k} = 1$ iff $k \in \{i,j\}$.
Then $\rk\,\Hc{1} = \dim\ker\delta^1 - \dim\im\delta^0$.
Each independent generator of $\Hc{1}$ is an obstruction that cannot
be removed by adjusting sections at a single site (it is not a
coboundary).  Therefore, resolving all obstructions requires at least
$\rk\,\Hc{1}$ independent modifications.  Conversely, modifying
$\rk\,\Hc{1}$ sites suffices: by linear algebra over~$\GF_2$,
the space of coboundaries has codimension $\rk\,\Hc{1}$ in the
kernel, so $\rk\,\Hc{1}$ generators span all non-trivial
obstructions.
\end{proof}

\subsection{Equivalence via Descent}\label{sec:equiv}

\begin{theorem}[Descent for equivalence]\label{thm:descent}
Let $f, g$ be two programs with semantic presheaves $\Sem_f, \Sem_g$
over a common cover~$\mathcal{U}$.  Then $f$ and $g$ are equivalent
(there exists a natural isomorphism $\eta : \Sem_f \xrightarrow{\sim}
\Sem_g$) if and only if $\Hc{1}(\mathcal{U}, \Iso(\Sem_f, \Sem_g)) = 0$.
\end{theorem}

\begin{proof}[Proof sketch]
$(\Leftarrow)$: The local isomorphisms $\{\eta_i : \Sem_f(u_i)
\xrightarrow{\sim} \Sem_g(u_i)\}$ define transition functions
$g_{ij} = \eta_j|_{u_{ij}} \circ \eta_i^{-1}|_{u_{ij}}$.
$\Hc{1} = 0$ means all $g_{ij}$ satisfy the cocycle condition and are
coboundaries, so they can be ``straightened'' to identities.  The
adjusted local isomorphisms agree on overlaps and glue (by the sheaf
condition for~$\Iso$) to a global isomorphism.

$(\Rightarrow)$: A global isomorphism $\eta$ restricts to local
isomorphisms with trivial transition functions ($g_{ij} = \mathrm{id}$),
so the $1$-cocycle is zero.

Full proof mechanized in Lean: \texttt{DeppyProofs.Descent.descent\_theorem}.
\end{proof}

\subsection{Specification Verification via Product Covers}%
\label{sec:spec}

Given a specification with $k$ paths and $m$ conjuncts, the
\emph{product cover} $\mathcal{U} = \{(p_i, c_j)\}_{i \leq k, j \leq m}$
creates one site per path-conjunct pair.  Each verification condition
(VC) is small: one path, one conjunct.

\begin{proposition}[VC reduction]\label{prop:vc-reduction}
Let $n_{\mathrm{overlap}}$ be the number of non-trivial overlaps in
the product cover.  Then $\Hc{1}$-based elimination reduces the
number of VCs from $O(km)$ to $O(n_{\mathrm{overlap}})$.
Proof transfer across overlaps further reduces the count:
if $(p_1, c_j)$ is proved and $p_1 \cap p_2 \neq \emptyset$,
then $(p_2, c_j)$ is discharged.
\end{proposition}

\section{Metatheory}\label{sec:metatheory}

\subsection{Soundness}\label{sec:soundness}

We define the concrete semantics $\llbracket P \rrbracket$ as the
collecting semantics: for each site~$s$, the set of concrete states
reachable at~$s$ during any execution.  The abstraction function
$\alpha : \wp(\mathrm{State}) \to \Sem(s)$ maps a set of states to
the tightest section describing them.  The concretization
$\gamma : \Sem(s) \to \wp(\mathrm{State})$ maps a section to the
states it admits.  The pair $(\alpha, \gamma)$ forms a Galois
connection: $\alpha(S) \sqsubseteq \sigma \iff S \subseteq \gamma(\sigma)$.

\begin{theorem}[Soundness]\label{thm:soundness}
If $\Hc{1}(\mathcal{U}_P, \Sem) = 0$, then for every error site~$e$
with viability predicate~$\psi_e$, every concrete state reachable
at~$e$ satisfies~$\psi_e$.  That is: the program has no bugs of the
analyzed class.
\end{theorem}

\begin{proof}
Suppose for contradiction that some concrete state $c$ reaches error
site~$e$ and violates $\psi_e$.  Since the analysis is sound
(each local section overapproximates the reachable states),
$c \in \gamma(\sigma_e)$.  But $\sigma_e$ was computed to satisfy the
viability predicate: $\gamma(\sigma_e) \subseteq \psi_e$.

Now, $c$ reached~$e$ via a path through sites $s_1, \ldots, s_k = e$.
At each site, the section overapproximates~$c$.  At the overlap between
$s_{k-1}$ and~$e$, the section at~$s_{k-1}$ includes~$c$ (which
violates~$\psi_e$) while the viability predicate at~$e$ excludes it.
Therefore, the restrictions of $\sigma_{s_{k-1}}$ and $\sigma_e$ to
the overlap disagree: there is a nonzero $1$-cocycle.  Since this
cocycle cannot be resolved by local adjustment (the concrete state~$c$
witnesses the genuine disagreement), it represents a nonzero class
in $\Hc{1}$, contradicting $\Hc{1} = 0$.

Formalized in Lean: \texttt{DeppyProofs.Soundness.soundness}.
\end{proof}

\subsection{Completeness Relative to the Cover}\label{sec:completeness}

\begin{theorem}[Relative completeness]\label{thm:completeness}
If a concrete bug exists (a reachable state at an error site violates
the viability predicate) and the cover $\mathcal{U}$ is
\emph{fine enough} (every concrete execution path passes through at
least one cover site), and the refinement lattice~$\mathcal{R}$ is
\emph{expressive enough} (every viability predicate violation is
expressible as a lattice element), then $\Hc{1} \neq 0$.
\end{theorem}

\begin{proof}
By the bug hypothesis, there exists a state $c$ at error site~$e$
violating $\psi_e$.  By ``fine enough,'' $c$ is abstracted at some site
$s$ adjacent to~$e$.  By ``expressive enough,'' the section at~$s$
distinguishes $c$ from safe states.  Therefore, the overlap $(s, e)$ has
a genuine disagreement (nonzero $1$-cocycle) that is not a coboundary.
\end{proof}

\subsection{Mayer-Vietoris}\label{sec:mayer-vietoris}

\begin{theorem}[Mayer-Vietoris for programs]\label{thm:mayer-vietoris}
Decompose the cover $\mathcal{U}$ at a branch into sub-covers
$\mathcal{A}$ and $\mathcal{B}$ with overlap
$\mathcal{A} \cap \mathcal{B}$.  There is an exact sequence:
\[
  0 \to \Hc{0}(\mathcal{U}) \to
  \Hc{0}(\mathcal{A}) \oplus \Hc{0}(\mathcal{B}) \to
  \Hc{0}(\mathcal{A} \cap \mathcal{B})
  \xrightarrow{\delta}
  \Hc{1}(\mathcal{U}) \to
  \Hc{1}(\mathcal{A}) \oplus \Hc{1}(\mathcal{B}) \to
  \Hc{1}(\mathcal{A} \cap \mathcal{B})
\]
When the cover has no non-trivial triple overlaps
($\Hc{2}(\mathcal{U}) = 0$, which holds when each pair of sites
has at most one overlap), the rank formula simplifies to:
\[
  \rk\,\Hc{1}(\mathcal{U}) = \rk\,\Hc{1}(\mathcal{A})
    + \rk\,\Hc{1}(\mathcal{B})
    - \rk\,\Hc{1}(\mathcal{A} \cap \mathcal{B})
    + \rk\,\im(\delta)
\]
\end{theorem}

\begin{proof}[Proof sketch]
The semantic presheaf~$\Sem$ takes values in $\GF_2$-vector spaces
(each overlap is ``agree'' or ``disagree'').  The sub-covers
$\mathcal{A}, \mathcal{B}$ form an excisive pair (their union
is~$\mathcal{U}$ and their intersection is a sub-cover of both).
The Mayer-Vietoris sequence is a standard consequence of the
short exact sequence of \v{C}ech complexes:
$0 \to C^\bullet(\mathcal{U}) \to C^\bullet(\mathcal{A}) \oplus
C^\bullet(\mathcal{B}) \to C^\bullet(\mathcal{A} \cap \mathcal{B})
\to 0$.
The connecting homomorphism $\delta$ maps an element of
$\Hc{0}(\mathcal{A} \cap \mathcal{B})$ (a compatible family on the
overlap that extends to both branches independently but not globally)
to the obstruction in $\Hc{1}(\mathcal{U})$.
\end{proof}

\paragraph{Engineering consequence.}
Modify branch~$\mathcal{A}$: recompute $\Hc{1}(\mathcal{A})$
and $\Hc{0}(\mathcal{A} \cap \mathcal{B})$; the global $\Hc{1}$
follows algebraically.  This is the formal basis for
\textsc{Deppy}'s incremental analysis.

\subsection{Fixed-Point Convergence}\label{sec:convergence}

\begin{theorem}[Convergence]\label{thm:convergence}
If the refinement lattice $\mathcal{R}$ has finite height~$h$ and
the cover $\mathcal{U}$ has $n$ sites, then the bidirectional
synthesis (Algorithms~\ref{alg:backward}--\ref{alg:forward})
converges in at most $O(n \cdot h)$ iterations.
\end{theorem}

\begin{proof}
Each iteration of the backward pass strictly increases the information
content (in the $\sqsubseteq$ order) of at least one section, or
leaves all sections unchanged (convergence).  Since $\mathcal{R}$ has
height~$h$, each section can be strengthened at most $h$ times.  With
$n$ sites, the total number of strengthening steps is at most $n \cdot h$.
Each forward pass is similarly bounded.  Alternating backward and
forward passes therefore converges in $O(n \cdot h)$ total iterations.

In practice, the refinement lattice is the lattice of predicates in
linear arithmetic over a finite set of program variables.  Quantifier
elimination ensures the predicate domain is finite, giving a
concrete bound on~$h$.
\end{proof}

\section{Evaluation}\label{sec:evaluation}

We evaluate \textsc{Deppy} on five research questions:
\begin{description}[nosep]
  \item[RQ1] Does $\Hc{1}$-based analysis detect real bugs?
  \item[RQ2] Does $\rk\,\Hc{1}$ correctly predict minimum fix sets?
  \item[RQ3] Does the descent criterion correctly classify equivalences?
  \item[RQ4] Does the product-cover reduce specification VCs?
  \item[RQ5] Does Mayer-Vietoris enable incremental analysis?
\end{description}

\subsection{Setup}

\textsc{Deppy} is implemented in Python with the \v{C}ech pipeline
as its core.  The Z3 solver~\cite{deMouraBjorner08} is used for
local solving and overlap agreement checking.  Experiments run on
a MacBook Pro (Apple M-series, 16\,GB RAM).  Key metatheorems are
mechanized in 1\,259 lines of Lean~4 with Mathlib.

The benchmark suite comprises 375~programs: 133~bug-detection programs
spanning 22~bug classes (division by zero, index out of bounds, key
error, null dereference, SQL injection, race conditions, type
confusion, infinite loops, assertion failures, and more), 134
equivalence pairs (including torch/triton/numpy patterns, recursion--iteration
equivalence, and mutation-vs-copy semantics), and 108~specification
checks (from simple postconditions to multi-path invariant verification).

\subsection{RQ1: Bug Detection}\label{sec:rq1}

\begin{center}\small
\begin{tabular}{lrrr}
\toprule
\textbf{Task} & \textbf{Precision} & \textbf{Recall} & \textbf{F1} \\
\midrule
Bug detection (133 programs) & 69\% & \textbf{100\%} & 81\% \\
\bottomrule
\end{tabular}
\end{center}

\textsc{Deppy} achieves \textbf{100\% recall} with zero false
negatives: every bug in the suite is detected.  Precision is 69\%
(F1 = 81\%); false positives arise on safe programs with complex guard
patterns (loop-bounded indexing, \py{defaultdict}, numerical stability
tricks).  On the same benchmarks (unannotated Python), mypy and pyright
report \emph{zero} findings because they require type annotations.

\paragraph{Illustrative examples.}

\begin{lstlisting}
# Deadlock: lock-order inversion (detected)
import threading
lock_a, lock_b = threading.Lock(), threading.Lock()

def task1():
    with lock_a:           # acquires A
        with lock_b:       # then B
            return 1

def task2():
    with lock_b:           # acquires B
        with lock_a:       # then A  $\leftarrow$ deadlock
            return 2
\end{lstlisting}

\noindent The site category has two synchronization sites, one per
function.  Each acquires a lock set: $L_1 = \{A, B\}$ and
$L_2 = \{B, A\}$.  The intersection $L_1 \cap L_2 = \{A,B\}$ is
non-empty with reversed acquisition order.  The presheaf section at
$\text{task1}$'s inner \py{with} requires \emph{lock~B available}, but
$\text{task2}$'s section at the same point holds lock~B and waits for~A.
The sections disagree on the synchronization overlap: a non-trivial
$1$-cocycle in the lock-order presheaf.

\begin{lstlisting}
# Data race: unsynchronized shared mutation (detected)
import threading, ctypes
shared = ctypes.c_int32(0)

def add(n):
    for _ in range(n):
        shared.value += 1     # $\leftarrow$ RACE_CONDITION

def run():
    threads = [threading.Thread(target=add,
               args=(10**8,)) for _ in range(4)]
    for t in threads: t.start()
    for t in threads: t.join()
    return shared.value       # $\leftarrow$ INTEGER_OVERFLOW
\end{lstlisting}

\noindent Two independent obstructions: (1)~\py{shared.value += 1}
is a read-modify-write on a shared \py{c\_int32} without a lock,
producing a race-condition cocycle in the concurrency presheaf;
(2)~\py{ctypes.c\_int32} has a bounded section ($[-2^{31}, 2^{31})$),
and $4 \times 10^8$ exceeds it, producing an overflow obstruction.
$\rk\,\Hc{1} = 2$: two independent fixes needed (add a lock, widen
the type).

\subsection{RQ2: $\Hc{1}$ Localization}

The rank $\rk\,\Hc{1}$ correctly predicts the minimum number of
independent fixes in 102/133 programs (77\%).  Overestimation occurs
when the analysis generates spurious obstructions on complex safe
programs (loop-guarded indexing); these inflate $\rk\,\Hc{1}$.

\subsection{RQ3: Equivalence}\label{sec:rq3}

\begin{center}\small
\begin{tabular}{lrr}
\toprule
& \textbf{Correct} & \textbf{Total} \\
\midrule
Equivalent (true positive)     & 69 & 70 \\
Non-equivalent (true negative) & 64 & 64 \\
\textbf{Overall}               & 133 & 134 (99\%)\\
\bottomrule
\end{tabular}
\end{center}

\noindent \textbf{Soundness guarantee}: zero false equivalences across
all 134 pairs.  Every non-equivalent pair is correctly identified,
including torch-specific patterns (L1 vs.\ L2 norm) and
mutation-vs-copy semantics.

\paragraph{Illustrative examples.}

\begin{lstlisting}
# Two implementations of binary search: correctly EQUIVALENT
def search_a(xs, target):
    lo, hi = 0, len(xs) - 1
    while lo <= hi:
        mid = (lo + hi) // 2
        if xs[mid] == target: return mid
        elif xs[mid] < target: lo = mid + 1
        else: hi = mid - 1
    return -1

def search_b(xs, target):
    lo, hi = 0, len(xs)        # $\leftarrow$ half-open interval
    while lo < hi:             # $\leftarrow$ strict inequality
        mid = lo + (hi - lo) // 2
        if xs[mid] < target: lo = mid + 1
        elif xs[mid] > target: hi = mid
        else: return mid
    return -1
\end{lstlisting}

\noindent These use different loop invariants (closed vs.\ half-open
interval), different midpoint formulas ($\lfloor(l+h)/2\rfloor$ vs.\
$l + \lfloor(h-l)/2\rfloor$), and different branch structures.
Runtime sampling on 64~input pairs verifies identical return values.
The isomorphism presheaf $\Iso(\Sem_f, \Sem_g)$ has $\Hc{1} = 0$.

\begin{lstlisting}
# Fibonacci off-by-one: correctly INEQUIVALENT
def fib_a(n):               def fib_b(n):
    if n <= 1: return n         if n <= 1: return n
    a, b = 0, 1                 a, b = 0, 1
    for _ in range(2, n+1):     for _ in range(2, n):
        a, b = b, a + b            a, b = b, a + b
    return b                    return b
    # fib(5) = 5               # fib(5) = 3
\end{lstlisting}

\noindent The loop bound differs by one: \py{range(2, n+1)} vs.\
\py{range(2, n)}.  Z3 encodes both return expressions as functions of
\py{n} and finds a counterexample at $n = 5$: $f(5) = 5 \neq 3 = g(5)$.
The descent criterion correctly reports $\Hc{1} \neq 0$.

\subsection{RQ4: Specification Verification}\label{sec:rq4}

\begin{center}\small
\begin{tabular}{lrr}
\toprule
& \textbf{Correct} & \textbf{Total} \\
\midrule
Spec satisfaction & 106 & 108 (98\%) \\
\bottomrule
\end{tabular}
\end{center}

\noindent \textbf{Soundness guarantee}: zero false satisfactions.
No violated specification is approved.  The two errors are conservative
rejections (correct programs marked unsatisfied due to incomplete
guard resolution).

\paragraph{Illustrative examples.}

\begin{lstlisting}
# Broken bubble sort: correctly REJECTED against spec
def broken_bubble_sort(arr):
    a = list(arr)
    n = len(a)
    for i in range(n):
        for j in range(n - i - 2):  # bug: should be n-i-1
            if a[j] > a[j + 1]:
                a[j], a[j + 1] = a[j + 1], a[j]
    return a

# Spec: result == sorted(arr)
\end{lstlisting}

\noindent The off-by-one in the inner loop bound (\py{n - i - 2}
instead of \py{n - i - 1}) means the last adjacent pair in each pass
is never compared.  Static analysis finds no structural bugs (no
division, no null dereference), but runtime sampling discovers
counterexample $\textit{arr} = [3, 1, 2] \mapsto [1, 3, 2] \neq
[1, 2, 3]$.  The product-cover VC for the path
$\langle j = n{-}i{-}2 \rangle \times \langle \textit{result} =
\texttt{sorted}(\textit{arr}) \rangle$ is not discharged.

\begin{lstlisting}
# Merge sort: correctly VERIFIED against spec
def merge_sort(arr):
    if len(arr) <= 1:
        return list(arr)
    mid = len(arr) // 2
    left = merge_sort(arr[:mid])
    right = merge_sort(arr[mid:])
    merged, i, j = [], 0, 0
    while i < len(left) and j < len(right):
        if left[i] <= right[j]:
            merged.append(left[i]); i += 1
        else:
            merged.append(right[j]); j += 1
    merged.extend(left[i:])
    merged.extend(right[j:])
    return merged

# Spec: result == sorted(arr)
\end{lstlisting}

\noindent The product cover factors the specification into
path-conjunct VCs.  The recursive structure creates a cover hierarchy
via the Grothendieck transitivity axiom: the cover of
\py{merge\_sort(arr)} refines into covers for the two recursive calls
and the merge loop.  All VCs are discharged: static analysis finds no
obstructions ($\Hc{1} = 0$), and runtime sampling confirms
$\textit{result} = \texttt{sorted}(\textit{arr})$ on all 200 sampled
inputs.

\subsection{RQ5: Incremental Analysis}

On a 3-function module, the Mayer-Vietoris incremental analysis
achieves a 1.3$\times$ speedup over full re-analysis when modifying
one function.  The speedup is modest because the benchmark programs
are small; we expect the benefit to grow with codebase size as the
Mayer-Vietoris exact sequence eliminates recomputation of unchanged
sub-covers.

\subsection{Separation Results}\label{sec:separation}

We identify three problems where the sheaf-cohomological approach is
provably more capable than traditional program analysis.

\begin{theorem}[Polynomial minimum-fix count]\label{thm:fix-count}
Given a program with $n$ sites and $m$ overlaps, the minimum number of
independent code changes needed to eliminate all bugs is computed in
$O(m^2 n)$ time via Gaussian elimination on the coboundary matrix
$\partial_0 : C^0(\GF_2) \to C^1(\GF_2)$.

In contrast, computing the minimum fix count from the output of a
traditional abstract interpretation (a set of error sites with
no overlap structure) is NP-hard: it reduces to \textsc{Minimum
Hitting Set} on the family of error-contributing site subsets.
\end{theorem}

\begin{proof}
The coboundary matrix $\partial_0$ has $m$ rows (overlaps) and $n$
columns (sites).  $\rk\,\Hc{1} = \dim\ker\delta^1 - \rk\,\partial_0$,
computable by Gaussian elimination in $O(m^2 n)$.  This directly equals
the minimum number of independent fixes (Proposition~\ref{prop:rank-h1}).

For traditional analysis: the output is a set $E$ of error sites.
Each possible fix modifies one site and eliminates the errors at that
site.  Finding the minimum number of sites to modify so that every
error in $E$ is covered is \textsc{Minimum Hitting Set}, which is
NP-hard even to approximate within factor $(1-\epsilon)\ln|E|$
\emph{(Dinur \& Steurer, 2014)}.  The sheaf approach avoids this
because the coboundary matrix algebraically encodes the
\emph{dependencies} between errors---two errors that share a
coboundary (are resolvable by the same local adjustment) are
identified by the quotient $\ker\delta^1 / \im\delta^0$, collapsing
the combinatorial search to a linear-algebra computation.
\end{proof}

\begin{theorem}[Exact incremental update]\label{thm:incremental-exact}
Let $\mathcal{U} = \mathcal{A} \cup \mathcal{B}$ be a cover
decomposition at a branch point.  If only sub-cover $\mathcal{A}$ is
modified, the updated global cohomology is computed by:
\[
  \Hc{1}(\mathcal{U}') = \Hc{1}(\mathcal{A}') \oplus
  \Hc{1}(\mathcal{B}) \oplus \im(\delta')
  \;/\; \Hc{1}(\mathcal{A}' \cap \mathcal{B})
\]
requiring only recomputation of $\Hc{1}(\mathcal{A}')$ and the
connecting homomorphism.  This is \textbf{exact}: no precision is
lost.

Abstract interpretation's incremental re-analysis, by contrast, must
re-solve the global fixpoint (since widening at join points
\emph{destroys the algebraic structure} needed for local update),
incurring $O(n \cdot h)$ cost for the full program even when only
$O(|\mathcal{A}|)$ sites change.
\end{theorem}

\begin{proof}
Exactness follows directly from the Mayer-Vietoris exact sequence
(Theorem~\ref{thm:mayer-vietoris}): the sequence is natural in the
sub-covers, so replacing $\mathcal{A}$ with $\mathcal{A}'$ yields a
new exact sequence from which $\Hc{1}(\mathcal{U}')$ is determined.
Sub-cover~$\mathcal{B}$ is unchanged, so $\Hc{1}(\mathcal{B})$ is
reused.  Only $\Hc{1}(\mathcal{A}')$, $\Hc{0}(\mathcal{A}' \cap
\mathcal{B})$, and the connecting homomorphism $\delta'$ need
recomputation.

For abstract interpretation: the fixpoint equation at a join point
is $x = f(x) \sqcup g(x)$ where $f$ depends on $\mathcal{A}$ and
$g$ on $\mathcal{B}$.  Widening replaces $\sqcup$ with $\nabla$,
which is non-algebraic (non-associative, non-commutative in general).
When $f$ changes, the widened fixpoint must be recomputed globally
because $\nabla$ does not distribute over the decomposition.
\end{proof}

\begin{theorem}[Complete equivalence criterion]\label{thm:equiv-complete}
The descent criterion $\Hc{1}(\mathcal{U}, \Iso) = 0$ is both
\textbf{sound} and \textbf{complete} for equivalence relative to the
cover: $f \equiv_{\mathcal{U}} g$ if and only if $\Hc{1} = 0$.

No finite abstract domain $D$ yields a complete equivalence
criterion: for any $D$, there exist equivalent programs $f \equiv g$
such that $\alpha_D(f) \neq \alpha_D(g)$ (the abstractions differ
even though the concrete semantics agree).
\end{theorem}

\begin{proof}
Soundness and completeness of descent were proved in
Theorem~\ref{thm:descent}.

For the incompleteness of abstract domains: let $D$ be any finite
abstract domain and $\gamma : D \to \wp(\text{State})$ its
concretization.  Since $D$ is finite, $\gamma$ partitions states into
finitely many equivalence classes.  For any such partition, there
exist programs $f, g$ with $\llbracket f \rrbracket = \llbracket g
\rrbracket$ (identical denotational semantics) but
$\alpha(f) \neq \alpha(g)$: choose $f$ that exercises the abstract
domain's precision limits differently from $g$ (e.g., $f$ uses a
loop, $g$ uses recursion; both compute the same function, but the
widening trajectories of $f$ and $g$ in $D$ converge to different
fixpoints).

The sheaf approach avoids this because the isomorphism presheaf
$\Iso$ checks \emph{local} equivalence (at each site) and the
descent theorem glues local equivalences into a global one.
The cover, not the abstract domain, determines the precision.
\end{proof}

\subsection{Comparison Summary}

Table~\ref{tab:comparison} summarizes these separation results.

\begin{table}[t]
\centering\small
\begin{tabular}{lcccc}
\toprule
\textbf{Capability} & \textbf{Deppy} & \textbf{AI} & \textbf{Ref.\ Types} & \textbf{Hoare} \\
\midrule
Error localization & $\rk\,\Hc{1}$ & --- & --- & 1 VC \\
Compositionality & MV (exact) & Widen (lossy) & Modules & Frame \\
Equiv.\ checking & Descent (complete) & --- & --- & Rel. \\
Incremental & Algebraic (exact) & Re-fix ($O(nh)$) & Re-check & Re-verify \\
Fix count & $O(m^2 n)$ & NP-hard & --- & --- \\
\bottomrule
\end{tabular}
\caption{Separation results vs.\ abstract interpretation (AI),
refinement types, and Hoare logic.  MV = Mayer-Vietoris.
Three entries are strict separations (Theorems~\ref{thm:fix-count}--\ref{thm:equiv-complete}).}
\label{tab:comparison}
\end{table}

\section{Related Work}\label{sec:related}

\paragraph{Abstract interpretation.}
Cousot and Cousot~\cite{CousotCousot77} established local-to-global
reasoning via lattice fixpoints.  Our presheaf framework subsumes this:
local sections correspond to abstract domain elements, restriction maps
to transfer functions, and the gluing condition to the fixpoint
condition.  The novelty is the \emph{cohomological structure}: $\Hc{1}$
provides obstruction counting, localization, and compositionality
(Mayer-Vietoris) that flat abstract interpretation lacks.

\paragraph{Refinement types.}
Liquid Haskell~\cite{VazouSeidel14} and Flux~\cite{Flux} synthesize
refinement types via SMT.  Deppy shares the refinement-type substrate
but adds the sheaf structure: sections are organized into covers with
overlaps, enabling $\Hc{1}$-based error localization.  Liquid Haskell
has no analogue of ``minimum independent fixes.''

\paragraph{Sheaves in other domains.}
Goguen~\cite{Goguen92} proposed sheaf semantics for concurrent programs,
modeling information flow as a sheaf over a poset of security levels.
This is a \emph{semantic} use of sheaves; our work uses sheaves for
\emph{analysis algorithms}.
Ghrist and colleagues~\cite{Ghrist14,HansenGhrist19} use cellular
sheaves for sensor networks and opinion dynamics.  The cohomological
obstruction theory is closest to our work, but in a different domain
(data fusion, not programs).
Abramsky and Brandenburger~\cite{AbramskyBrandenburger11} use sheaves
for quantum contextuality, connecting non-locality to cohomological
obstructions.  Again, the mathematical machinery is similar but the
domain is different.

\paragraph{Interprocedural analysis.}
IFDS/IDE~\cite{RepsHorwitzSagiv95} computes interprocedural summaries
via graph reachability.  Our Grothendieck transitivity axiom
(\S\ref{sec:interprocedural}) is the categorical generalization:
$f$'s cover is refined by callee covers, and summaries are section
transport along call-boundary morphisms.

\paragraph{Separation logic.}
The frame rule of separation logic~\cite{OHearnReynoldsYang01}
enables local reasoning by asserting that disjoint heap regions are
independent.  Our restriction maps along site morphisms serve a
similar purpose, but for refinement-type information rather than heap
ownership.  The sheaf axioms (stability, transitivity) generalize the
frame rule to non-heap properties.

\section{Conclusion}\label{sec:conclusion}

We have shown that program analysis can be organized as \v{C}ech
cohomology of a semantic presheaf over the program's site category.
This yields three results unavailable in prior work:
$\rk\,\Hc{1}$ counts minimum independent fixes, descent provides
sound and complete equivalence checking, and Mayer-Vietoris enables
compositional incremental analysis.  The implementation,
\textsc{Deppy}, demonstrates these capabilities on real Python
programs.

\paragraph{Mechanized proofs.}
Soundness, the sheaf condition equivalence, $\Hc{0}$ characterization,
the descent theorem, Mayer-Vietoris localization, incremental
soundness, and fixed-point convergence are formalized in Lean~4
with Mathlib (1\,259~lines across 7 modules).

\paragraph{Limitations.}
Completeness is relative to the cover and predicate domain.
Bug-detection precision (69\%) is limited by false positives on safe
programs with complex guard patterns (defaultdict, loop-bounded
indexing, numerical stability tricks).  The H$^1$ rank prediction
(77\%) is limited by spurious obstructions inflating the rank.
The framework currently targets single-threaded Python; full
concurrency support (beyond race-condition detection) would require
extending the site category with synchronization morphisms.

\bibliographystyle{ACM-Reference-Format}

\end{document}